\documentclass[letterpaper,10 pt, conference]{ieeeconf}
\IEEEoverridecommandlockouts
\usepackage{cite}
\usepackage{amsmath,amssymb,amsfonts,amsthm}
\usepackage{graphicx}
\usepackage{textcomp}
\usepackage{xcolor}
\usepackage{cleveref}
\usepackage[ruled,vlined]{algorithm2e}
\def\BibTeX{{\rm B\kern-.05em{\sc i\kern-.025em b}\kern-.08em
    T\kern-.1667em\lower.7ex\hbox{E}\kern-.125emX}}
\usepackage[export]{adjustbox}

\theoremstyle{definition}

\newtheorem{theorem}{Theorem}
\newtheorem{lemma}{Lemma}

\newtheorem{assumption}{Assumption}

\newtheorem*{remark}{Remark}
\title{\LARGE \bf Adversarial attacks in consensus-based multi-agent reinforcement learning}
\author{Martin Figura, Krishna Chaitanya Kosaraju, and Vijay Gupta
\thanks{M. Figura, K. C. Kosaraju, and V. Gupta are with the Department of Electrical Engineering at the University of Notre Dame, Notre Dame, IN, 46556 USA,
\{{\tt mfigura,kkosaraj,vgupta2}\}@nd.edu.}
}
\begin{document}
\maketitle
\thispagestyle{empty}
\pagestyle{empty}

\begin{abstract}
Recently, many cooperative distributed multi-agent reinforcement learning (MARL) algorithms have been proposed in the literature. In this work, we study the effect of adversarial attacks on a network that employs a consensus-based MARL algorithm. We show that an adversarial agent can persuade all the other agents in the network to implement policies that optimize an objective that it desires. In this sense, the standard consensus-based MARL algorithms are fragile to attacks. 
\end{abstract}
\section{Introduction}
Multi-agent reinforcement learning (MARL) is a  branch of reinforcement learning (RL) \cite{sutton1998}, where multiple decision-makers learn a policy that is optimal in the context of competitive, cooperative, or mixed objectives \cite{zhang2019,busoniu2008}. A recent success story of MARL in popular parlance include its performance in the video game Starcraft II by training an agent which outperforms the world's best human players \cite{vinyals2019}. In this paper, we focus on MARL for cooperative agents. Potential applications in this stream have been proposed for long for diverse fields such as sensor networks \cite{cortes2004}, robotics \cite{yang2004}, and traffic control \cite{kuyer2008}.  

Early formulations of MARL assumed that the agents share a common reward and focused on decentralized decision-making. The sample efficiency in the training of a multi-agent network in this setting is significantly improved by establishing communication between agents \cite{foerster2016}. Nonetheless, the centralized reward is often infeasible due to overwhelming communication requirements and complex network topology, which motivated the development of cooperative distributed MARL with decentralized knowledge of rewards. In this setting, each agent has a local utility function and views only its own reward. The problem is still cooperative in the sense that the agents wish to maximize the sum of all utility functions. In this problem, the agents must communicate not only to improve the sample efficiency but also to become aware of the other agents' performance. Only by propagating information over the entire network, the agents can achieve a common objective, e.g., to maximize team-average returns. The ability to learn a policy that maximizes the common objective in a partially observable environment can be facilitated by consensus algorithms as presented in \cite{zhang2018}, whereby the agent's rewards (and possibly actions) remain unknown to the rest of the network and must not be directly communicated between agents to ensure their privacy.\par

Consensus algorithms are generally devised for distributed systems to find agreement on signal values over networks \cite{olfati2007}. These algorithms find applications in many fields including sensor networks \cite{olfati2005consensus}, coordination of vehicles \cite{ren2007}, or even blockchain \cite{mingxiao2017}. In practice, consensus algorithms must be robust to faults that arise from relatively frequent occurrences of interrupted communication links or corrupted signals \cite{fischer1983}. Therefore, the convergence of resilient consensus algorithms was rigorously studied under different considerations for the nature of adversarial attacks \cite{leblanc2013}, graph topology \cite{saldana2017,sundaram2007}, or frequency of communication \cite{ding2016}. These research efforts naturally complement studies of the influence of adversarial attacks on network performance. A simple yet powerful result from the analysis of linear consensus \cite{olfati2005consensus} states that the topology of a consensus matrix determines the limiting value for the consensus updates. In the presence of a single \textit{malicious} agent, which does not apply consensus updates, the limiting value coincides with the adversary's value.\par

In the consensus MARL algorithm in \cite[Algorithm~2]{zhang2018}, every agent estimates the team-average reward and value function using linear approximations and exchanges parameters with other agents through a consensus protocol. Interestingly, this scheme guarantees the asymptotic convergence to the team-average optimal policy even with simultaneous actor, critic, and consensus updates over time-varying communication graphs. Furthermore, the algorithm retains the convergence property even with sparse data transmission for strongly connected graphs \cite{lin2019}. 

In this paper, we study the effects of adversarial attacks on the consensus MARL algorithm \cite[Algorithm~2]{zhang2018} with discounted rewards in the objective function. The attacks we consider are different from the commonly studied data poisoning attacks in ML or RL, which seek to understand if changing the data or rewards by an external agent can degrade the performance of RL algorithms \cite{ma2019}. Here, we consider a MARL setting where a participating agent itself is malicious. Specifically, we ask whether a single adversarial agent can either prevent convergence of the algorithm, or even worse, lead the other agents to optimize a utility function that it chooses. We show that the answer to this question is in the affirmative by designing a suitable attack and analyzing the convergence of the algorithm under it. 

Specifically, we take under the scope networks with a single malicious agent, i.e., with an adversary that can compromise the consensus and critic updates and transmits the same signal values to its neighbors. We show that when the malicious agent greedily attempts to maximize its own well-defined objective function, all other agents in the network end up maximizing the adversary's objective function as well. We provide a proof of asymptotic convergence analogous to \cite{zhang2018}. Our study motivates the development of resilient consensus MARL algorithms.

The paper is structured as follows. We provide a background of the networked Markov decision process along with the agents' objectives in Section 2. In Section 3, we state all assumptions in a compact form, present the consensus MARL algorithm, and provide the convergence analysis. Section 4 is dedicated to numerical simulations.

\paragraph*{Notation}
We let $\mathbf{1}$ denote the vector of ones. The operator $\otimes$ represents the Kronecker product. The cardinality of a set $\mathcal{C}$ is denoted by $|\mathcal{C}|$.
\section{Background}
\subsection{Networked Markov decision process}
We consider a networked Markov decision process (MDP) given as a tuple  $(\mathcal{S},\{\mathcal{A}^i\}_{i\in\mathcal{N}},\mathcal{P},\{\mathcal{R}^i\}_{i\in\mathcal{N}},\mathcal{G})$, where $\mathcal{N}=\{1,\dots,N\}$, $\mathcal{S} $ is a set of states, $\mathcal{P}$ is a set of transitional probabilities, $\gamma\in[0,1)$ is a discount factor, $\mathcal{G}$ represents a graph, and $\mathcal{A}^i$ and $\mathcal{R}^i$ are a set of actions and rewards of agent $i$, respectively. The graph $\mathcal{G}=(\mathcal{N},\mathcal{E})$ is defined by a set of vertices $\mathcal{N}$ associated with the agents in the network and a set of edges $\mathcal{E}\subseteq\mathcal{N}\times\mathcal{N}$. The global state and action are denoted by $s$ and $a$, respectively, $s^\prime$ denotes the state at the next time step, and the superscript $i$ denotes a signal of agent $i$. We let $r^i(s,a,s^\prime):\mathcal{S}\times\mathcal{A}\times\mathcal{S}\rightarrow\mathcal{R}_i\subset\mathbb{R}$ denote the individual reward of subsystem $i$, $p(s^\prime|s,a):\mathcal{S}\times\mathcal{S}\times\mathcal{A}\rightarrow\mathcal{P}\subset\mathbb{R}$ the joint transitional probability, and $\pi^i(a^i|s):\mathcal{S}\times\mathcal{A}_i\rightarrow(0,1)$ the policy of subsystem~$i$. The overall network policy can be written as a stacked vector of individual policies, $\pi(a|s)=[\pi^1(a^1|s)^T,\dots,\pi^N(a^N|s)^T]^T$. In case we need to explicitly specify a signal value at time~$t$, we use subscript $t$, i.e., $r_{t+1}^i(s_t,a_t,s_{t+1})$. An important consideration in this work is that agent $i$, $i\in\mathcal{N}$, receives the private reward $r_{t+1}^i$ along with the observation of the global state $s_t$ and action $a_t$ at each step in training.\par
We let $\pi(a|s;\theta)=[\pi^1(a^1|s;\theta^1)^T,\dots,\pi^N(a^N|s;\theta^N)^T]^T$ denote the network policy parameterized by $\theta$, where $\pi^i(a^i|s;\theta^i)$ is a parameterized policy of agent $i$. Further, we distinguish between reward signals by making the following definitions:
\begin{enumerate}
	\item average individual reward at $(s,a)$:\\ $$\hat{r}^i(s,a)=\sum_{s^\prime}p(s^\prime|s,a)r^i(s,a,s^\prime)$$
	\item average individual reward under global policy $\pi(a|s;\theta)$: $$\hat{r}_\theta^i(s)=\sum_a\pi(a|s;\theta)\hat{r}^i(s,a)$$
		\item average individual reward at all state-action pairs $(s,a)\in\mathcal{S}\times\mathcal{A}$: $$\hat{R}^i=[\hat{r}^i(s,a),s\in\mathcal{S},a\in\mathcal{A}]^T\in\mathbb{R}^{|\mathcal{S}|\cdot|\mathcal{A}|}$$
	\item average individual reward under global policy $\pi(a|s;\theta)$ at all states $s\in\mathcal{S}$: $$\hat{R}_\theta^i=[\hat{r}_\theta^i(s),s\in\mathcal{S}]^T\in\mathbb{R}^{|\mathcal{S}|}.$$

\end{enumerate}
We also define team-average rewards $r(s,a,s^\prime)=\frac{1}{N}\sum_{i=1}^Nr^i(s,a,s^\prime)$, $\hat{r}(s,a)=\frac{1}{N}\sum_{i=1}^N\hat{r}^i(s,a)$, $\hat{r}_\theta(s)=\frac{1}{N}\sum_{i=1}^N\hat{r}_\theta^i(s)$, $\hat{R}_\theta=\frac{1}{N}\sum_{i=1}^N\hat{R}^i_\theta\in\mathbb{R}^{|\mathcal{S}|}$, and $\hat{R}=\frac{1}{N}\sum_{i=1}^N\hat{R}^i$. Furthermore, we define the estimated network reward function at $(s,a)$ as $\bar{r}(s,a;\lambda)$, where $\lambda$ are the function parameters.

\subsection{Objective}
We let $\mathcal{N}^+$ and $\mathcal{N}^-$ denote the set of cooperative agents and adversaries, respectively, and note that $\mathcal{N}=\mathcal{N}^+\cup\mathcal{N}^-$. The objective of agents $i\in\mathcal{N}^+$ is to maximize a team-average objective function given as
\begin{align}\label{obj_coop}
    \max_\theta J^+(\theta)=\max_\theta\mathbb{E}\big[\sum_{t=0}^\infty\frac{1}{N}\sum_{i=1}^N\gamma^tr_{t+1}^i|s_0=s\big].
\end{align}
The cooperative agents are unaware of the presence of an adversarial agent that
seeks to maximize a different objective function. We define the objective function for $i\in\mathcal{N}^-$ as
\begin{align}\label{obj_greedy}
    \max_\theta J^-(\theta)=\max_\theta\mathbb{E}\big[\sum_{t=0}^\infty \gamma^tr_{t+1}^{i}|s_0=s\big].
\end{align}
It is important to note that the adversarial agent can compromise the rewards $r_{t+1}^i$, $i\in\mathcal{N}^-$, to incentivize its malicious behavior. Furthermore, once the agents establish communication the adversary can spread false information about the performance of the entire network embedded in the compromised rewards $r_{t+1}^i$. This may eventually lead to incentivizing a bad behavior in the cooperative agents as well, regardless of whether the maximized objective is \eqref{obj_coop} or \eqref{obj_greedy}. In Section~III, we show that the entire network maximizes the adversarial agents' objective in \eqref{obj_greedy} when the adversarial agent ignores signals transmitted by the cooperative agents.

\section{Multi-agent actor-critic algorithm under adversarial attacks}

In this section, we present assumptions on the network and signals, define the consensus MARL algorithm, state main theorems concerning the convergence of the actor and critic, and prove that the adversarial agent persuades the remaining agents in the network to maximize its individual objective in \eqref{obj_greedy} despite their initial agreement to maximize the team objective in \eqref{obj_coop}.\par
We formally define true action-value functions of the cooperative agents and the adversary under policy $\pi(a|s)$ in the respective order as follows
\begin{align*}
Q_\pi^{i}(s,a)&=\mathbb{E}_\pi\bigg[\sum_{t=0}^\infty\frac{1}{N}\gamma^t\sum_{k\in\mathcal{N}}r_{t+1}^k\bigg],\,i\in\mathcal{N}^+\\
Q_\pi^{i}(s,a)&=\mathbb{E}_\pi\bigg[\sum_{t=0}^\infty \gamma^tr_{t+1}^{i}\bigg],\,i\in\mathcal{N}^-
\end{align*}
where $j\in\mathcal{N}^-$. Further, we define the true state value functions as
\begin{align*}
V^i_\pi(s)=\sum_a \pi(a|s)Q_\pi^i(s,a),\quad i\in\mathcal{N}.
\end{align*}
We will approximate $V^i_\pi(s)$ using a parameterized state-value function $V(s;v^i)$.
\begin{remark}
It is required that all agents use the same basis functions (or neural networks) so that their parameter values $v^i$ can eventually converge to a consensus value. This limitation can be overcome by considering gossip-based algorithms \cite{mathkar2016rl}. However, the convergence analysis for gossip-based algorithms is rather challenging.
\end{remark}

\subsection{Assumptions}
In this subsection, we state assumptions needed for the convergence of the consensus MARL algorithm which we introduce later in the next section. The assumptions are similar to \cite{zhang2018}. 
\begin{assumption}
The policy $\pi^i(a^i|s;\theta^i)>0$ for any $i\in\mathcal{N}$, $\theta^i\in\Theta^i$, $s\in\mathcal{S}$, $a^i\in\mathcal{A}^i$. Also, $\pi^i(a^i|s;\theta^i)$ is continuously differentiable with respect to $\theta^i$. For any $\theta\in\Theta$, we let $P_\theta(s_{t+1}|s_t)=\sum_{a_t\in\mathcal{A}}P(s_{t+1}|s_t,a_t)\pi(a|s;\theta)$ denote the transition matrix of the Markov chain $\{s_t\}_{t\geq 0}$ induced by policy $\pi(a|s;\theta)$. The Markov chain $\{s_t\}_{t\geq 0}$ is irreducible and aperiodic under any $\pi(a|s;\theta)$.
\end{assumption}
\begin{assumption}
The update of the policy parameter $\theta_t^i$ includes a local projection operator, $\Gamma_i:\mathbb{R}^{m_i}\rightarrow\mathbb{R}^{m_i}$, that projects any $\theta_t^i$ onto a compact set $\Theta^i$. Also, we assume that $\Theta=\Pi_i^N\Theta^i$ is large enough to include at least one local minimum of $J(\theta)$.
\end{assumption}
\begin{assumption}
The instantaneous reward $r_{t+1}^i(s_t,a_t,s_{t+1})$ is uniformly bounded for any $i\in\mathcal{N}$ and $t\geq 0$.
\end{assumption}
\begin{assumption}
The sequence of random matrices $\{C_t\}_{t\geq 0}\subseteq\mathbb{R}^{N\times N}\subseteq\mathbb{R}^{N\times N}$ satisfies
\begin{enumerate}
	\item $C_t$ is row stochastic, i.e., $C_t\mathbf{1} = \mathbf{1}$, and $c_t(i,j)=1$ for $i=j\in\mathcal{N}^-$.
	There exists a constant $\eta\in(0,1)$ such that, for any $c_t(i,j)>0$, we have $c_t(i,j)\geq\eta$.
	\item $C_t$ respects the communication graph $\mathcal{G}_t$, i.e., $c_t(i,j)=0$ if $(i,j)\notin\mathcal{E}_t$.
	\item The spectral norm of $\mathbb{E}[C_t^T(I-\mathbf{11}^T/N)C_t]$ belongs to $[0,1)$.
	\item Given the $\sigma$-algebra generated by the random variables before time $t$, $C_t$ is conditionally
independent of $r_{t+1}^i$ for any $i\in\mathcal{N}$.
\end{enumerate}
\end{assumption}
\begin{assumption}
For each agent $i$, the state-value function and the team reward function are both parameterized by linear functions, i.e., $V (s;v) = v^T\varphi(s)$ and $\bar{r}(s,a;\lambda)=\lambda^Tf(s,a)$, where $\varphi(s) = [\varphi_1(s),\dots,\varphi_L(s)]\in\mathbb{R}^L$ and $f(s,a)=[f_1(s,a),\dots,f_M(s,a)]\in\mathbb{R}^M$ are the features associated with $s$ and $(s,a)$, respectively. The feature vectors $\varphi(s)$ and $f(s,a)$ are uniformly bounded for any $s\in\mathcal{S}$, $a\in\mathcal{A}$. Furthermore, let the feature matrix $\Phi\in\mathbb{R}^{|\mathcal{S}|\times L}$ have $[\varphi_l(s), s\in\mathcal{S}]^T$ as its $l$-th column for any $l\in [L]$, and the feature matrix $F\in\mathbb{R}^{|\mathcal{S}|\cdot|\mathcal{A}|\times M}$ have $[f_m(s,a), s\in\mathcal{S},a\in\mathcal{A}]^T$ as its $m$-th column for any $m\in[M]$. Both $\Phi$ and $F$ have full column rank.
\end{assumption}

\begin{assumption}
The step sizes $\alpha_{v,t}$ and $\alpha_{\theta,t}$ satisfy $\sum_t\alpha_{v,t}=\infty$, $\sum_t\alpha_{\theta,t}=\infty$, $\sum_t \alpha^2_{v,t}+\alpha^2_{\theta,t}<\infty$, $\alpha_{\theta,t}=o(\alpha_{v,t})$, and $\lim_{t\rightarrow\infty}\alpha_{v,t+1}\alpha_{v,t}^{-1}=1$.
\end{assumption}

\begin{assumption}
The set $\mathcal{N}^-$ contains exactly one element, i.e., there is one malicious agent with a well-defined objective specified in \eqref{obj_greedy}.
\end{assumption}

We note that Assumption~4.3 is satisfied when the communication graph $\mathcal{G}$ is connected in the mean. To simplify the convergence analysis in Section III.C, we assume that there is only one adversary that is learning using compromised rewards and does not perform consensus updates. The latter leads to unbalanced consensus updates in the entire network. We note that more general adversarial attacks are possible, e.g., there may be multiple adversarial agents in the network that may perform arbitrary parameter updates. Nonetheless, the narrow scope of attacks presented in this work is sufficient to demonstrate the fragility of the vanilla consensus MARL algorithm.
\subsection{MARL algorithm}\label{alg}
In this subsection, we introduce the consensus MARL algorithm. We let $\Delta^i$ denote an estimated network TD error estimated by agent $i$. We noted earlier in Section~III that every agent maintains parameters $v^i$ which describe the network value function approximation $V(s,v^i)$. Further, we recall that the rewards $r^i(s,a,s^\prime)$ remain private but the agents are allowed to estimate the network reward function. Intuitively, estimating the network reward function is a necessary step since the agents try to maximize the team-average objective in \eqref{obj_coop}. We let $d_\theta(s)$ denote a stationary distribution of the Markov chain $\{s_t\}_{t\geq 0}$ under policy $\pi(a|s;\theta)$. If the rewards were mutually observable among the agents, they would minimize the weighted mean square error
\begin{align}
arg\min_\lambda\sum_{s\in\mathcal{S},a\in\mathcal{A}}d_\theta(s)\pi(a|s;\theta)\big[\bar{r}(s,a;\lambda)-\hat{r}(s,a)\big]^2.\label{opt_cen}
\end{align}
The optimization problem in \eqref{opt_cen} can be recast into a distributed optimization problem, which has the same stationary points, given as follows
\begin{align}
arg\min_\lambda\frac{1}{N}\sum_{i=1}^N\sum_{s\in\mathcal{S},a\in\mathcal{A}}d_\theta(s)\pi(a|s;\theta)\big[\bar{r}-\hat{r}^i\big]^2.
\end{align}
since $\frac{1}{N}\sum_i\big(\bar{r}(s,a;\lambda)-\hat{r}^i(s,a)\big)=\bar{r}(s,a;\lambda)-\hat{r}(s,a)$. Hence, the agents can individually perform gradient steps to update the parameters $\lambda^i$ based on their true rewards $r^i(s,a,s^\prime)$. By communicating via a consensus protocol, they further gain information about the encoded rewards of the other agents. The estimation and communication of the network reward function parameters $\lambda^i$ and network value function parameters $v^i$ provide the agents with the ability to update their policy to benefit the team. The consensus actor-critic algorithm, a version of \cite[Algorithm~2]{zhang2018} with discounted returns is given in Algorithm~\ref{algorithm}. We note that the algorithm is the same for all agents but the adversary omits the consensus step. Furthermore, the action taken by the full network $a_t$ can be assumed unobservable if the estimated rewards are independent of the actions, i.e., $\bar{r}(s,a;\lambda^i)=\bar{r}(s;\lambda^i)$. In the following subsection, we provide the convergence analysis for Algorithm~1.
\begin{remark}
The scope of this work can be easily extended to \cite[Algorithm~1]{zhang2018}, where agents approximate state-action value function parameters. For such an algorithm, the global action $a_t$ must be observable by all agents in the network.
\end{remark}

\begin{algorithm}[t]\label{algorithm}
	\SetAlgoLined
	 \textbf{Initialize parameters} $\theta_0^i,\lambda^i_0,\tilde\lambda^i_0,v_0^i,\tilde{v}^i_0$, $\forall i\in\mathcal{N}$\;
	 \textbf{Initialize} $s_0,\{\alpha_{v,t}\}_{t\geq 0},\{\alpha_{\theta,t}\}_{t\geq 0}, t\leftarrow 0$\;
	\textbf{Repeat until convergence} \\
 	\For{$i\in\mathcal{N}$}{
  		\textbf{Observe} state $s_{t+1}$, action $a_t$, and reward $r_{t+1}^i$\;
  		\textbf{Update}\\
  		$\tilde\lambda_{t}^i\leftarrow \lambda_t^i+\alpha_{v,t}\big(r_{t+1}^i-\bar{r}_{t+1}(\lambda_t^i)\big)\nabla_{\lambda}\bar{r}_{t+1}(\lambda_t^i)$\;
  		$\delta_t^i\leftarrow r_{t+1}^i+\gamma V(s_{t+1};v_t^i)-V(s_t;v_t^i)$\;
  		$\Delta_t^i\leftarrow \bar{r}_{t+1}(\lambda_t^i)+\gamma V(s_{t+1};v_t^i)-V(s_t;v_t^i)$\;
  		\textbf{Update critic} $\tilde v_{t}^i\leftarrow v_t^i+\alpha_{v,t}\delta_t^i\nabla_v V(s_t;v_t^i)$\;
  		\textbf{Update actor} $\theta_{t+1}^i\leftarrow\theta_t^i+\alpha_{\theta,t}\Delta_t^i\psi_t^i$\;
  		\textbf{Send} $\tilde\lambda_t^i,\tilde{v}_t^i$ to the neighbors over $\mathcal{G}_t$\;
  		\textbf{Take} action $a_{t+1}^i\sim\pi^i(a_{t+1}^i|s_{t+1};\theta_{t+1}^i)$\;
  		 }
  	\For{$i\in\mathcal{N}$}{
  		\textbf{Consensus step}\\
  		$\lambda_{t+1}^i=\sum_{j\in\mathcal{N}} c_t(i,j)\tilde\lambda_{t}^j,\quad v_{t+1}^i=\sum_j c_t(i,j)\tilde v_t^j$\;
  	}
  	\textbf{Update} iteration counter $t\leftarrow t+1$
 \caption{Consensus MARL algorithm}
\end{algorithm}

\subsection{Convergence analysis}
In this subsection, we proceed with the convergence analysis. First, we show that the critic $V(s;v^i)$ and network reward function $\bar{r}(s,a;\lambda^i)$ converge to a fixed point for all $i\in\mathcal{N}$ while the policy $\pi(a|s;\theta)$ remains fixed. Then, we prove the full convergence of the actor updates that occur on a slower timescale. We write the stationary distribution of Markov chain $d_\theta(s)$ for all states $s\in\mathcal{S}$ as a matrix $D_\theta^s=diag[d_\theta(s),s\in\mathcal{S}]$. Similarly, we write the distribution of all state-action pairs $(s,a)$ as a matrix $D_\theta^{s,a}=diag[d_\theta(s)\cdot\pi(a|s;\theta),s\in\mathcal{S},a\in\mathcal{A}]$. For brevity, we use shorthands $\varphi_t=\varphi(s_t)$ and $f_t=f(s_t,a_t)$. Finally, we define the consensus value $\left<z_t\right>=\frac{1}{N}\sum_i z_t^i$ and the disagreement vector $z_{\perp,t}=z_t-\mathbf{1}\otimes\left<z_t\right>$.

\begin{theorem}
Under assumptions 1 and 3-7, for any policy $\pi(a|s;\theta)$, with the updates of $\{v_t^i\}$ in Algorithm~1, we have $\lim_tv_t^i=v_\theta$ and $\lim_t\lambda_t^i=\lambda_\theta$ for $i\in\mathcal{N}$  almost surely. Furthermore, $v_\theta$ and $\lambda_\theta$ are the unique solutions to
\begin{align}
F^TD_\theta^{s,a}(\hat{R}^{j}-F\lambda_\theta)&=0\label{lambda_equilibrium}\\
\Phi^TD_\theta^s\big(\hat{R}_\theta^{j}+\gamma P_\theta\Phi v_\theta-\Phi v_\theta\big)&=0,\label{v_equilibrium}
\end{align}
where $j\in\mathcal{N}^-$.
\end{theorem}
\begin{proof}

We let $z_t=[(z_t^1)^T,\dots,(z_t^N)^T]^T\in\mathbb{R}^{(M+L)N}$, where $z_t^i=[(\lambda_t^i)^T,(v_t^i)^T]^T$. Furthermore, we define $b_t=r_{t+1}\otimes\begin{bmatrix}f_t^T & \phi_t^T\end{bmatrix}^T$ and $A_t=I\otimes A_t^\prime$, where $A_t^\prime\begin{bmatrix}-f_tf_t^T & 0 \\0 & \phi_t(\gamma\phi_{t+1}-\phi_t)^T \end{bmatrix}$. We let $\mathcal{F}_t^z=\{z_0,Y_\tau,\tau\leq t\}$ denote a filtration where $Y_\tau=\{s_\tau,a_\tau,r_\tau,C_{\tau-1}\}$ is a collection of random variables.
The iterations of Algorithm~1 can be written in a compact form as follows
\begin{align*}
z_{t+1}&=(C_t\otimes I)\big(z_t+\alpha_{v,t}(A_tz_t+b_t)\big)\\
&=(C_t\otimes I)\big[z_t+\alpha_{v,t}\big(h(z_t,Y_t)+M_{t+1}\big)\big]
\end{align*}
where $h(z_t,Y_t)=\mathbb{E}(A_tz_t+b_t|\mathcal{F}_t^z)$ and $M_t=A_tz_t+b_t-\mathbb{E}(A_tz_t+b_t|\mathcal{F}_t^z)$. To prove Theorem~1, we need to show that
\begin{enumerate}
	\item \textbf{Lemma 1}: the parameters $\lambda_t$ and $v_t$ remain bounded for all $t\geq0$,
	\item \textbf{Lemma 2}: the adversary's parameters asymptotically converge, i.e., $\lambda_t^j\rightarrow\lambda_\theta$ and $v_t^j\rightarrow v_\theta$, $j\in\mathcal{N}^-$,
	\item \textbf{Lemma 3}: the agents' parameters asymptotically converge to the consensus value $\left<\lambda_t\right>$ and $\left<v_t\right>$.
\end{enumerate}
We take advantage of the rich convergence analysis in \cite{zhang2018} to prove the lemmas.
\begin{lemma}
	Under assumptions 1 and 3-6, the sequence $\{z_t\}$ satisfies $\sup_t||z_t||<\infty$ almost surely.
\end{lemma}
\begin{proof}
The proof is given in \cite[Appendix C]{zhang2018}. The only difference in our work is that in the absence of the consensus step the updates of $z_t^i$, $i\in\mathcal{N}$, asymptotically follow the ODE $\dot{z}_t^i=\bar{A}_t^\prime z_t^i+\bar{b}_t^i$  where
\begin{align}
	\bar{A}_t^\prime&=\begin{bmatrix}-F^TD_\theta^{s,a}F & 0 \\ 0 &\Phi^TD_\theta^s(\gamma P_\theta-I)\Phi\end{bmatrix}\label{A_t}\\
	\bar{b}_t^i&=\begin{bmatrix}(F^TD_\theta^{s,a}\hat{R}^i)^T & (\Phi^TD_\theta^s\hat{R}_\theta^i)^T\end{bmatrix}^T.\label{b_t}
\end{align}
The discount factor satisfies $\gamma\in[0,1)$ and the stochastic matrix $P_\theta$ has positive eigenvalues that are less than or equal to 1. Therefore, the matrix $\Phi^TD_\theta^{s,a}(\gamma P_\theta-I)\Phi$ has eigenvalues with strictly negative real parts, which implies that the ODE $\dot{z}_t^i=\bar{A}_t^\prime z_t^i+\bar{b}_t^i$ has an asymptotically stable equilibrium. Hence, $\sup_t||z_t||<\infty$ almost surely.

\end{proof}
\begin{lemma}
Under assumptions 1, 3, and 5-7, $\lim_{t\rightarrow\infty} z_t^j=z_\theta$, $j\in\mathcal{N}^-$, almost surely. Furthermore,  $z_\theta=[\lambda_\theta^T,v_\theta^T]^T$ is a unique solution to \eqref{lambda_equilibrium} and \eqref{v_equilibrium}.

\end{lemma}
\begin{proof}
We recall that the adversarial agent does not perform the consensus step. Using the findings in Lemma 1, we can immediately conclude that $\dot{z}_t^j=\bar{A}_t^\prime z_t^j+\bar{b}_t^j$ is the limiting ODE, with $A_t^\prime$ given in \eqref{A_t} and $\bar{b}_t^j=\begin{bmatrix}(F^TD_\theta^{s,a}\hat{R}^j)^T & (\Phi^TD_\theta^s\hat{R}_\theta^j)^T\end{bmatrix}^T$. The ODE has a unique asymptotically stable equilibrium $z_\theta=[\lambda_\theta^T,v_\theta^T]^T$ that satisfies \eqref{lambda_equilibrium} and \eqref{v_equilibrium}.
\end{proof}
\begin{lemma}[Appendix B.4, Step 1 in \cite{zhang2018}]
Under assumptions 1 and 3-7, the disagreement vector $z_{\perp,t}$ satisfies $\lim_{t\rightarrow\infty}z_{\perp,t}=0$ almost surely.
\end{lemma}
To complete the proof of Theorem~1, we recall that
\begin{itemize}
	\item $\lim_{t\rightarrow\infty} (z_t^j-z_\theta)=0$ for $j\in\mathcal{N}^-$ a.s. (Lemma 2)
	\item $\lim_{t\rightarrow\infty}(z_t^i-\left<z_t\right>)=0$ for $i\in\mathcal{N}$ a.s. (Lemma 3).
\end{itemize}
Therefore, $\lim_{t\rightarrow\infty}(\left<z_t\right>-z_\theta)=0$ almost surely where $z_\theta=[\lambda_\theta^T,v_\theta^T]^T$ satisfies \eqref{lambda_equilibrium} and \eqref{v_equilibrium}.
\end{proof}

Having proved the critic and network reward convergence under a fixed policy $\pi(a|s;\theta)$, we proceed to make a statement about the convergence of the actor updates on the slower timescale.
\begin{theorem}\cite[Theorem 4.10]{zhang2018}
Under assumptions 1-7, the policy parameter $\theta_t^i$ converges almost surely to a point in the set of asymptotically stable equilibria of
\begin{align}
\dot\theta^i&=\hat\Gamma^i\big[\mathbb{E}_{s_t\sim d_\theta,a_t\sim\pi_\theta}(\Delta_{t,\theta}\cdot\psi_{t,\theta}^i)\big]\quad \text{for}\quad i\in\mathcal{N},
\end{align}
where $\mathbb{E}_{s_t\sim d_\theta,a_t\sim\pi_\theta}(\Delta_{t,\theta}\cdot\psi_{t,\theta}^i)=\mathbb{E}_{s_t\sim d_\theta,a_t\sim\pi_\theta}\big[(f_t^T\lambda_\theta+\gamma\varphi_{t+1}v_\theta-\varphi_tv_\theta)\nabla_{\theta^i}\log\pi^i(a_t^i|s_t;\theta^i)\big]$.
\end{theorem}
We note that the policy $\pi(a|s;\theta)$ converges to an equilibrium where the estimated network TD error $\Delta_{t,\theta}$ is equal to zero. Since $\Delta_{t,\theta}$ is a function of the parameterized network reward function $\bar{r}(s,a;\lambda)$ and value function $V(s;v)$, the policy $\pi(a|s;\theta)$ does not converge to the true optimal policy. The error between $\pi(a|s;\theta)$ and the true optimal policy that maximizes \eqref{obj_greedy} can be reduced by selecting appropriate models for $\bar{r}(s,a;\lambda)$ and $V(s;v)$.\par
In the next section, we present an example in which a group of agents employs Algorithm~1 to learn an optimal policy and is subject to a malicious attack from a single adversarial agent.

\section{Numerical simulations}
\begin{figure}[b]
\centering
\includegraphics[width=1\linewidth]{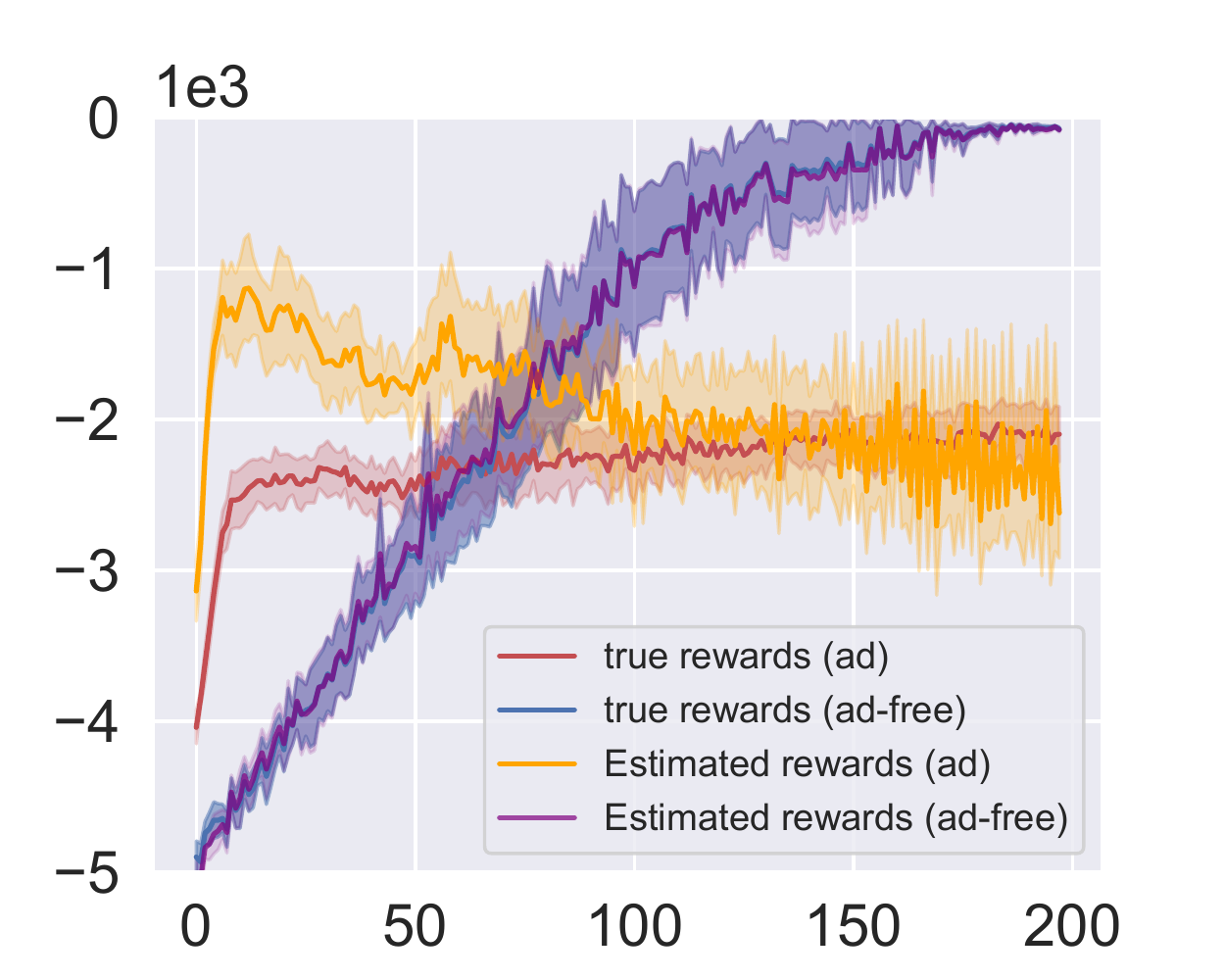}
\caption{Cumulative team-average rewards per episode for the adversary-free and attacked network. The comparison shows that the former performs significantly better}
\label{fig:average_returns}
\end{figure}

In this section, we assess the performance of Algorithm~1 through numerical simulations using nonlinear function approximation. We also compare our results against the decentralized actor-critic algorithm in \cite{zhang2018}. The code for these experiments can be found in \cite{git_marl_2020}.

We consider a network of agents $\mathcal{N}=\{1,2,3,4,5\}$ in a grid-world of dimension $(6\times 6)$. The position of agent~$i$ is described by the tuple $(x_i, y_i)\in S^i$, where $\mathcal{S}^i = [0,~\ldots,~5]^2$. We note that the tuples $(0,0)$ and $(5~,5)$ correspond to the top-left and bottom-right corners of the grid, respectively. The state of the grid-world is given as $s = [(x^i,~y^i),~i \in\mathcal{N}] \in \mathcal{S}$ where $\mathcal{S} = \mathcal{S}^1\times\ldots\times \mathcal{S}^5$. The cardinality of $S$ is $|\mathcal{S}|=36^5$ ($\approx$ 60.5 million states). Agent~$i$, $i\in\mathcal{N}$, takes actions from the set $\mathcal{A}^i=\{ 0:Left, ~1:Right, ~2:Up, ~3:Down, ~4:Stay\}$. If an action is to bring the agent to an infeasible state, then the agent remains in the same state. The set of actions of the network is given as $\mathcal{A}= \mathcal{A}^1 \times\ldots\times \mathcal{A}^5$, whose cardinality is $|\mathcal{A}|=5^5$.\par
The goal of each cooperative agent, $i\in\mathcal{N}^+=\{2,3,4,5\}$, is to maximize the objective in \eqref{obj_coop}, whereas the adversary $i\in\mathcal{N}^-=\{1\}$ attempts to maximize \eqref{obj_greedy}. The rewards of agent $i$, $i\in\mathcal{N}$, are given as follows
\begin{align*}
    r^i(s^i)&= -|x^i-x_{\mathrm{des}}^i|-|y^i-y_{\mathrm{des}}^i| - q^i,
\end{align*}
where $q^i$ denotes the number of neighboring agents that agent $i$ collides at the current time step. For simplicity, we consider that the communication graph $\mathcal{G}$ is fully connected and the consensus matrix $C$ in the adversary-free scenario has elements $c(i,j)=1/5$ for $i\in\mathcal{N}^+,\,j\in\mathcal{N}$.

\begin{figure}[t]
\centering
\includegraphics[width=1\linewidth]{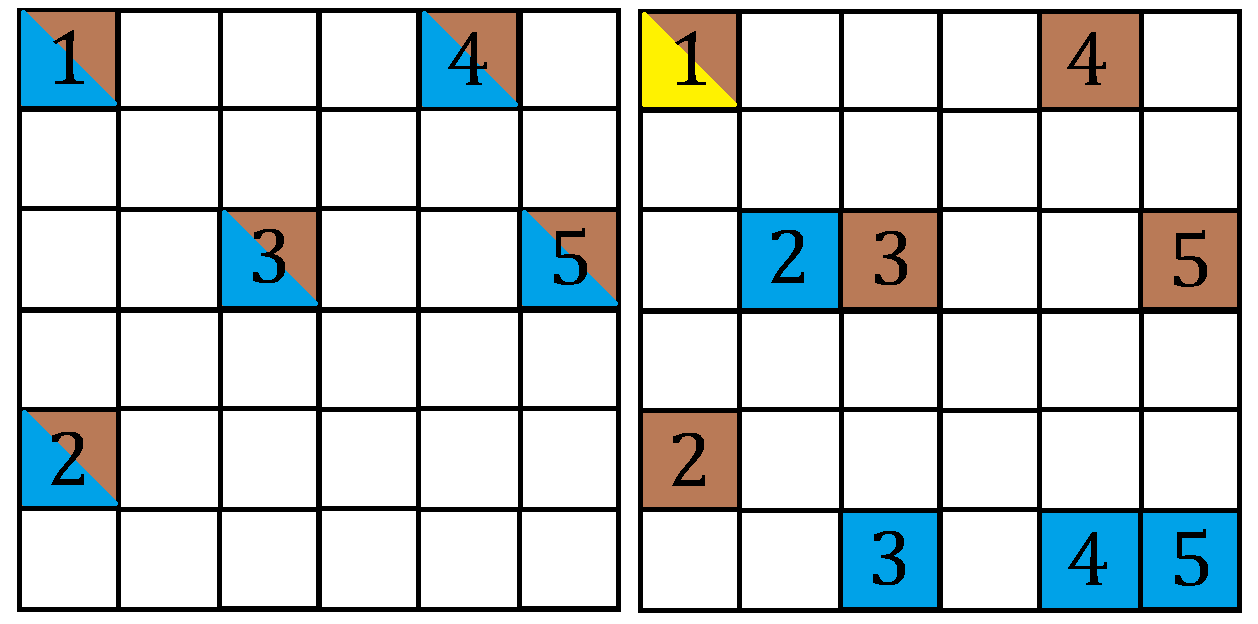}
\caption{Final network states in a simulation with no adversary (left) and with an adversary (right) after training for 200 episodes. Blue, yellow, and brown cells correspond to the cooperative agents', adversary's, and desired positions, respectively. All agents reach their desired positions when the network is adversary-free, whereas only the adversary finds its desired position when it attacks the network.}
\label{fig:end_states}
\end{figure}

\begin{figure*}[t]
    \centering
    \includegraphics[trim={6cm 0 6cm 0},clip, width=17cm, height=4cm]{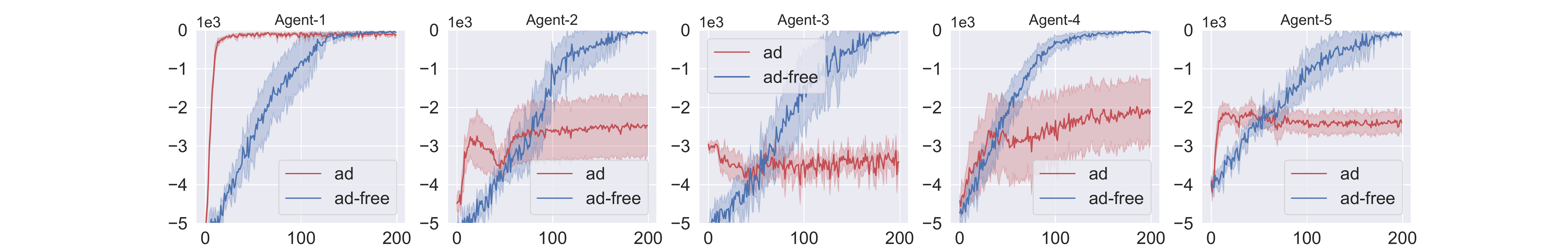}
    \caption{True cumulative rewards per episode obtained by each agent in the network. Blue and red color depict the performance of adversary-free network and the attacked network, respectively. The adversary quickly learns an optimal policy because it acts greedily with respect to the rest of the network. }
    \label{fig:agent_returns}
\end{figure*}

In both scenarios, we trained the agents for 200 episodes. Final states of a simulation of the grid-world after training are shown in Fig.~\ref{fig:end_states}. The agents' positions in the grid-world were randomly initialized in each training episode that was set to terminate after $1000$ steps or when the agents have reached their desired positions. The actor $\pi^i(a^i|s,\theta^i)$, critic $V(s,v^i)$, and global reward functions $\bar{r}(s,\lambda^i)$ were approximated using artificial neural networks with two hidden layers. In Fig.~\ref{fig:average_returns}, we compare the true cumulative team-average returns and cumulative estimated rewards obtained by the network in each episode. The estimated reward function converged in both scenarios but the convergence rate was slower in the presence of the adversary. The accumulated rewards per episode of each agent are depicted in Fig.~\ref{fig:agent_returns}. We can see that all agents learn a near-optimal policy when there is no adversary in the network. The adversary indeed harms the network, i.e., it learns a near-optimal policy but the remaining agents in the network perform poorly compared to the adversary-free scenario.

\section{Conclusion and future work}
In this paper, we showed that the general consensus MARL algorithm originally proposed in \cite{zhang2018} is not robust to adversarial attacks. We studied a well-defined malicious attack whereby a single adversarial agent attempts to compromise the objective function of a network of agents. We showed in the analysis that the network policy upon convergence locally maximizes the adversary's objective function under this specific malicious attack. Our work naturally raises the question of whether we can develop consensus-based MARL algorithms that are resilient to general adversarial attacks. Such attacks may include compromised rewards, arbitrary parameter updates, and arbitrary changes in the policy. There are many results on resilient consensus algorithms in the literature but it is unclear if the theoretical analysis can carry over to RL algorithms. The unique challenge for resilient consensus MARL is to provide robustness for the functions jointly estimated by the network of agents while the rewards remain private.
\bibliography{references}
\bibliographystyle{ieeetr}
\end{document}